\documentclass[twocolumn,pra,floatfix,amsmath,superscriptaddress]{revtex4-1}
\usepackage[latin9]{inputenc}
\usepackage{graphicx}
\usepackage[hypertex]{hyperref}
\usepackage{bm,bbm}
\usepackage{amsmath}
\usepackage{amsfonts}
\usepackage{amsthm}
\usepackage{amssymb}
\usepackage{mathrsfs}
\usepackage{subfigure}
\usepackage{array}
\usepackage{color}
\makeatletter

\usepackage{pxfonts}

\newtheorem{lemma}{Lemma}

\begin{document}

\title{Comment on ``Unified view of quantum correlations and quantum coherence''}

\author{Yao Yao}
\email{yaoyao@mtrc.ac.cn}
\affiliation{Microsystems and Terahertz Research Center, China Academy of Engineering Physics, Chengdu Sichuan 610200, China}
\affiliation{Institute of Electronic Engineering, China Academy of Engineering Physics, Mianyang Sichuan 621999, China}

\author{Li Ge}
\affiliation{School of Science, Hangzhou Dianzi University, Hangzhou Zhejiang 310018, China}

\author{Mo Li}
\email{limo@mtrc.ac.cn}
\affiliation{Microsystems and Terahertz Research Center, China Academy of Engineering Physics, Chengdu Sichuan 610200, China}
\affiliation{Institute of Electronic Engineering, China Academy of Engineering Physics, Mianyang Sichuan 621999, China}

\date{\today}

\begin{abstract}
We show that contrary to what it is claimed in Phys. Rev. A \textbf{94}, 022329 (2016), in general the local projective
measurement that induces maximal coherence loss is \textit{not} the projection onto the local basis that defines the coherence of
the system, at least for all quantum-incoherent states.
\end{abstract}


\maketitle

For a single-partite system $\rho=\sum_{i,j}\rho_{i,j}|i\rangle\langle j|$ with a reference (coherence-defining) basis $\{|i\rangle\}$,
it is obvious that the measurement that maximally eliminates the coherence in the system is the projective measurement
$\Pi(\rho)=\sum_i|i\rangle\langle i|\rho|i\rangle\langle i|$. Motivated by this intuition,
in the Appendix of Ref. \cite{Tan2016}, the authors attempted to prove that the local projective
measurement that induces maximal coherence loss is the projection onto the local basis that defines the coherence of
the system, \textit{even} for bipartite or multipartite quantum states, which is formally stated in the following.

\textit{Proposition.} For any bipartite state $\rho_{AB}=\sum_{i,j,k,l}\rho_{i,j,k,l}|i,j\rangle\langle k,l|$ where the coherence is measured
with respect to the local reference bases $\{|i\rangle_A\}$ and $\{|j\rangle_B\}$, the projective measurement on subsystem $B$
that induces maximal coherence loss is $\Pi_B(\rho_{AB})=\sum_j(\openone_A\otimes|j\rangle_B\langle j|)\rho_{AB}(\openone_A\otimes|j\rangle_B\langle j|)$.

To review the (flawed) proof, here we adopt the same notations introduced in \cite{Tan2016}. By using the spectral decomposition of
$\rho=\sum_np_n|\psi^n\rangle_{AB}\langle\psi^n|$, the authors assumed the following matrix representation of $\rho_{AB}$ with respect to
the local reference bases $\{|i\rangle_A\}$ and $\{|j\rangle_B\}$
\begin{align}
\rho_{AB}=\sum_n\sum_{i,j,k,l}p_n\psi^n_{i,j}(\psi^n_{k,l})^\ast|i,j\rangle_{AB}\langle k,l|.
\end{align}
Note that the coherence of the system is measured with respect to these bases. Consider some complete basis $\{|\lambda_m\rangle\}$ on $B$, and
corresponding projective measurement
$\Pi_B(\rho_{AB})=\sum_m(\openone_A\otimes|\lambda_m\rangle_B\langle \lambda_m|)\rho_{AB}(\openone_A\otimes|\lambda_m\rangle_B\langle\lambda_m|)$.
By computing the matrix elements, the authors \textit{in fact} proved that the coherence of the post-measurement state $\Pi_B(\rho_{AB})$
is lower bounded by the coherence of the reduced state of subsystem $A$, namely
\begin{align}
\sum_{i,j,k,l}\left|[\Pi_B(\rho_{AB})]_{i,j,k,l}\right|\geq\sum_{i,k}\left|\sum_{n,p}p_n\psi^n_{i,p}(\psi^n_{k,p})^\ast\right|,
\label{inequality1}
\end{align}
where $\rho_A=\sum_{i,k}\sum_{n,j}p_n\psi^n_{i,j}(\psi^n_{k,j})^\ast|i\rangle\langle k|$. For comparison, when $|\lambda_j\rangle=|j\rangle$
the absolute sum of the elements of $\Pi_B(\rho_{AB})$ should be
\begin{align}
\sum_{i,j,k,l}\left|[\Pi_B(\rho_{AB})]_{i,j,k,l}\right|=\sum_{i,j,k}\left|\sum_np_n\psi^n_{i,j}(\psi^n_{k,j})^\ast\right|.
\label{inequality2}
\end{align}
Obviously, the right hand side of Eq. (\ref{inequality2}) is generally larger than (that is, inequivalent to) that of Eq. (\ref{inequality1}) and thus
the claim of the authors is not correct.

For a simple counterexample for the Proposition, we can evaluate the coherence of the following state
\begin{align}
\rho_{AB}=\frac{1}{2}|+\rangle_A\langle+|\otimes|0\rangle_B\langle0|+\frac{1}{2}|-\rangle_A\langle-|\otimes|1\rangle_B\langle1|,
\end{align}
where $|\pm\rangle=\frac{1}{\sqrt{2}}(|0\rangle\pm|1\rangle)$ and the computational basis is assumed to be the coherence-defining basis.
Now if we perform a local projective measurement in the computational basis of subsystem $B$, the coherence of the bipartite system is invariant
and no coherence loss occurs since the whole state keeps unchanged. However, if we perform a local projective measurement in the dual basis
$\{|\pm\rangle\}$, the coherence of the bipartite system is \textit{completely} eliminated since in this case
$\Pi_B(\rho_{AB})=\frac{1}{2}\openone_A\otimes\frac{1}{2}\openone_B$ and the identity operator $\openone_{A(B)}$ is incoherent in any basis.

Indeed, this example can be extended to arbitrary quantum-incoherent states \cite{Chitambar2016}, which are of the following form
\begin{align}
\chi_{AB}=\sum_ip_i\rho^A_i\otimes|i\rangle_B\langle i|,
\end{align}
where $\{|i\rangle_B\}$ is the incoherent basis for subsystem $B$. Inspired by the above example, we perform a local projective measurement
in a mutually unbiased basis $\{|\lambda_j\rangle_B\}$ with respect to $\{|i\rangle_B\}$ \cite{Bandyopadhyay2002,Durt2010}, which means
$|\langle i|\lambda_j\rangle|^2=\frac{1}{d_B}$ for all $i$ and $j$ ($d_B$ denotes the dimension of subsystem $B$).
In this circumstance, the post-measurement state is given as
\begin{align}
\Pi_B^{\{|\lambda_j\rangle\}}(\chi_{AB})=\sum_ip_i\rho^A_i\otimes\sum_j\frac{1}{d_B}|\lambda_j\rangle_B\langle\lambda_j|=\rho_A\otimes\frac{1}{d_B}\openone_B.
\end{align}
For any valid coherence measure $C(\bullet)$ defined in the framework of \cite{Baumgratz2014},
it is easy to see that
\begin{align}
C(\Pi_B^{\{|i\rangle\}}(\chi_{AB}))&=C(\chi_{AB})=\sum_ip_iC(\rho^A_i) \nonumber\\
&\geq C(\sum_ip_i\rho^A_i)=C(\rho_A)=C(\Pi_B^{\{|\lambda_j\rangle\}}(\chi_{AB})),
\label{inequality3}
\end{align}
where we have used the convexity of $C(\bullet)$ \cite{Baumgratz2014}. The above inequality
indicates that the maximal coherence loss is alternatively induced by a projective measurement
in an arbitrary mutually unbiased basis (see also Eq. (\ref{inequality1})), which
is obviously opposed to the Proposition.

Note that, in Eq. (\ref{inequality3}), we have used the following lemma.
\begin{lemma}
For any valid coherence measure, we have $C(\rho\otimes\sigma_{\textrm{inc}})=C(\rho)$, where
$\sigma_{\textrm{inc}}$ is an incoherent state.
\end{lemma}
\begin{proof}
Here we present two alternative proofs. First, we can adopt the methodology in Ref. \cite{Bu2017}, namely,
the dismissal quantum operation (partial trace) and the appending quantum operation (with an incoherent state)
are all incoherent operations. Therefore, by using the monotonicity of the coherence measures under
incoherent operations \cite{Baumgratz2014}, we have the inequality
\begin{align}
C(\rho)\geq C(\rho\otimes\sigma_{\textrm{inc}})\geq C(\rho),
\end{align}
which implies $C(\rho\otimes\sigma_{\textrm{inc}})=C(\rho)$. Alternatively, we can also employ the framework
proposed in \cite{Yu2017}, which is equivalent to that of \cite{Baumgratz2014}. Yu \textit{et al.} proved
that a valid coherence measure should satisfy the following condition: $C(p_1\rho_1\oplus p_2\rho_2)=
p_1C(\rho_1)+p_2C(\rho_2)$ for block-diagonal states $\rho$ in the incoherent basis. Thus,
since $\sigma_{\textrm{inc}}=\sum_ip_i|i\rangle\langle i|$, we have
\begin{align}
C(\rho\otimes\sigma_{\textrm{inc}})=C(\oplus_ip_i\rho)=\sum_ip_iC(\rho)=C(\rho),
\end{align}
which completes the proof.
\end{proof}

In conclusion, we have proved that the Proposition raised by the authors of Ref. \cite{Tan2016} is not valid, at least
for all quantum-incoherent states, where the projective measurements in mutually unbiased bases play a significant role.
Furthermore, we believe this problem is highly nontrivial and probably state-dependent.
A thorough solution to this problem is still left as an open question.

\begin{acknowledgments}
This research is supported by the Science Challenge Project (Grant No. TZ2017003)
and the National Natural Science Foundation of China (Grant No. 11605166).
\end{acknowledgments}


\end{document}